\documentclass[eng]{csam}

\usepackage{graphicx}
\usepackage{natbib}
\usepackage{bm}
\usepackage{amsmath,amssymb}
\usepackage{listings}
\usepackage{hyperref}

\newtheorem{prop}{Proposition}

\submit{draft}

\begin{document}

\title{A tutorial on generalizing the default Bayesian $t$-test via posterior sampling and encompassing priors}

\author{Thomas J. Faulkenberry\footnote{Corresponding author: Assistant Professor, Department of Psychological Sciences, Tarleton State University, Stephenville, Texas, USA. E-mail: faulkenberry@tarleton.edu}\address[a]{Department of Psychological Sciences, Tarleton State University}}

\begin{abstract}
  With the advent of so-called ``default'' Bayesian hypothesis tests, scientists in applied fields have gained access to a powerful and principled method for testing hypotheses. However, such default tests usually come with a compromise, requiring the analyst to accept a one-size-fits-all approach to hypothesis testing. Further, such tests may not have the flexibility to test problems the scientist really cares about.  In this tutorial, I demonstrate a flexible approach to generalizing one specific default test (the JZS $t$-test; Rouder et al., 2009) that is becoming increasingly popular in the social and behavioral sciences. The approach uses two results, the Savage-Dickey density ratio (Dickey \& Lientz, 1980) and the technique of encompassing priors (Klugkist et al., 2005) in combination with MCMC sampling via an easy-to-use probabilistic modeling package for R called Greta. Through a comprehensive mathematical description of the techniques as well as illustrative examples, the reader is presented with a general, flexible workflow that can be extended to solve problems relevant to his or her own work.

  \vspace{5mm}

Note: this paper is in press at {\it Communications for Statistical Applications and Methods}.
\end{abstract}



\section{Introduction}
The $t$-test is one of the simplest, yet most enduring, examples of a hypothesis test that the social and behavioral scientist uses in his or her daily work. In the typical framework of null hypothesis significance testing (NHST), the $t$-test works by first assuming a {\sl null hypothesis}, and then calculating a $t$-score, which indexes the likelihood of obtaining some sample of observed data under the null hypothesis. If this probability is small, the scientist rejects the null in favor of some alternative hypothesis.  

Consider the following scenario that is often used when assessing the effect of some treatment. Let $x_{i1}$ and $x_{i2}$ denote measurements for the $i^{th}$ participant in two different conditions (e.g., a pretest and posttest). Consider the difference $d_i=x_{i2}-x_{i1}$. A typical consideration is whether these differences are different from 0; answering this question in the affirmative would then imply that the treatment had some nonzero effect. To answer this question, one can apply the standard one-sample $t$-test, which works by first assuming
\[
  d_i \sim \text{Normal}(\mu,\sigma^2),
\]
and then defining two competing hypotheses: a null hypothesis $\mathcal{H}_0:\mu=0$ and an alternative hypothesis $\mathcal{H}_1:\mu\neq 0$. We then test $\mathcal{H}_0$ by computing
\[
  t=\frac{\overline{d}}{s/\sqrt{n}},
\]
where $\overline{d}$ is the mean of the differences $d_i$ across all participants $i=1,\dots,n$, $s$ is the sample standard deviation of the difference scores $d_i$, and $n$ is the sample size. Under the null $\mathcal{H}_0$, the distribution of $t$ is well known as {\sl Student's $t$ distribution}, with density function
\[
  f(x) = \frac{\Gamma\left(\frac{\nu + 1}{2}\right)}{\sqrt{\nu \pi} \Gamma\left(\frac{\nu}{2}\right)} \left(1+\frac{x^2}{\nu}\right)^{-\frac{\nu+1}{2}},
\]
where $\nu$ represents degrees of freedom, and $x\in(-\infty,\infty)$. The cumulative distribution function $F(x)=\int_{-\infty}^x f(u)du$ can then be used to index the probability of observing data at least as extreme as that which we observed under the null hypothesis $\mathcal{H}_0$. Specifically, we compute $p(|x|>t) = 1-F(t)$, a quantity commonly known as a $p$-value.  If this probability is small (say, less than 5\%), then one may decide to reject $\mathcal{H}_0$ and conclude that $\mu \neq 0$, thus implying that our treatment had some nonzero effect.

This idea is well known to practicing social and behavioral scientists. However, there are some properties of this procedure that are suboptimal for robust inference. For one, the procedure is asymmetric \citep{rouder2009}. Suppose that one calculates a $p$-value above some commonly-used threshold like 5\%. What decision does the researcher make? Surely the logical opposite of ``reject $\mathcal{H}_0$'' is ``accept $\mathcal{H}_0$''. However, this decision rule is inconsistent. The reason for this follows from examining the distribution of $p$-values that result from increasing sample sizes. When the null is false (i.e., $\mu_1\neq \mu_2$), the value of $t$ increases as sample sizes increase.  Thus, the probability of rejecting $\mathcal{H}_0$ increases accordingly. However, if the null is true, $p$-values are uniformly distributed between 0 and 1, regardless of sample size. So, whereas a false null hypothesis can always be rejected if sample size is large enough, a true null hypothesis is always susceptible to being incorrectly rejected. Such inconsistency leads to asymmetry in the testing procedure -- increasing sample size can increase evidence against a false null hypothesis, but there is no corresponding way to increase evidence for a true null hypothesis.

Another criticism of the traditional hypothesis testing procedure is that researchers often misinterpret the results of such tests. \citet{hoekstra2014} asked 562 researchers and students from the field of psychology to assess the validity of six different statements involving incorrect interpretations of confidence intervals (e.g., ``The probability that the true mean is greater than 0 is at least 95\%''). Although each of these statements was false, both students and researchers on average believed at least 3 of the statements were true. Furthermore, researchers did no better than students with respect to these misunderstandings. This finding echos results by \citet{oakes1986}, who performed a similar study using statements about $p$-values; see also \citet{gigerenzer2004}.

In light of these criticisms, the social and behavioral sciences have seen an increase in recommendations to find alternatives to the orthodox use of null hypothesis significance testing \citep{wagenmakers2011}. One such alternative is to use Bayesian inference, and in particular Bayes factors \citep{kass1995,raftery1995,masson2011}. Bayesian inference is based on calculating the {\sl posterior probability} of a hypothesis $\mathcal{H}$ after observing data $\bm{y}$. This calculation proceeds by Bayes' theorem, which states
\begin{equation}\label{eq:bayes}
  p(\mathcal{H}\mid \bm{y}) = \frac{p(\bm{y}\mid \mathcal{H}) \cdot p(\mathcal{H})}{p(\bm{y})}.
\end{equation}
One way to think of Equation \ref{eq:bayes} is as follows: prior to observing data, one assigns a prior belief $p(\mathcal{H})$ to a hypothesis $\mathcal{H}$. Once the data $\bm{y}$ have been observed, one updates this prior belief to a posterior belief $p(\mathcal{H}\mid \bm{y})$ by multiplying the prior $p(\mathcal{H})$ by the likelihood $p(\bm{y}\mid \mathcal{H})$. This product is then rescaled to meet the requirements for being probability distribution (i.e., total probability = 1) by dividing by $p(\bm{y})$, the marginal probability of the observed data averaged across all possible hypotheses $\mathcal{H}$.

While this computation is fundamentally quite basic, one immediate consequence is how it can be used for comparing two hypotheses. Suppose as above that we have two competing hypotheses: a null hypothesis $\mathcal{H}_0$ and an alternative hypothesis $\mathcal{H}_1$. We can directly compare our posterior beliefs in these two hypotheses by computing their ratio $p(\mathcal{H}_0\mid \bm{y}) / p(\mathcal{H}_1\mid \bm{y})$, which we call the {\sl posterior odds} for $\mathcal{H}_0$ over $\mathcal{H}_1$.  Using Bayes' theorem (Equation \ref{eq:bayes}), we can readily see
\begin{equation}\label{eq:odds}
  \underbrace{\frac{p(\mathcal{H}_0\mid \bm{y})}{p(\mathcal{H}_1\mid \bm{y})}}_{\text{posterior odds}} = \underbrace{\frac{p(\bm{y}\mid \mathcal{H}_0)}{p(\bm{y}\mid \mathcal{H}_1)}}_{\text{Bayes factor}} \cdot \underbrace{\frac{p(\mathcal{H}_0)}{p(\mathcal{H}_1)}}_{\text{prior odds}}.
\end{equation}
As with Bayes' theorem, Equation \ref{eq:odds} can also be interpreted in terms of an ``updating'' metaphor.  Specifically, the posterior odds ratio is equal to the prior odds ratio multiplied by an updating factor.  This updating factor is the ratio of the marginal likelihoods $p(\bm{y}\mid \mathcal{H}_0)$ and $p(\bm{y}\mid \mathcal{H}_1)$, and is called the Bayes factor \citep{jeffreys1961,kass1995}.  The Bayes factor is the weight of evidence provided by data $\bm{y}$.  For example, suppose that one assigned the prior odds of $\mathcal{H}_0$ to $\mathcal{H}_1$ to be equal to 4-to-1; that is, we believe that, {\sl a priori}, $\mathcal{H}_0$ is 4 times more likely to be true than $\mathcal{H}_1$.  Then, suppose that after observing data, we compute a Bayes factor was computed to be 5.  Now, the posterior odds (the odds of $\mathcal{H}_0$ over $\mathcal{H}_1$ {\it after} observing data) is 20-to-1 in favor of $\mathcal{H}_0$ over $\mathcal{H}_1$.  

There are two immediate advantages to using the Bayes factor for inference. First, the Bayes factor is a ratio, and thus, is subject to a natural interpretation. Simply put, larger is better - the bigger the Bayes factor, the bigger the weight of evidence provided by the observed data. Second, since there was no specific assumption about the order in which we addressed $\mathcal{H}_0$ and $\mathcal{H}_1$, we could have just as easily measured the weight of evidence in favor of $\mathcal{H}_1$ over $\mathcal{H}_0$. In fact, once we have a Bayes factor in favor of one hypothesis, a simple reciprocal will give us the Bayes factor in favor of the the other hypothesis. In our example above, the Bayes factor for $\mathcal{H}_1$ over $\mathcal{H}_0$ would have been 1/5, implying that the data would actually {\sl decrease} our relative belief in $\mathcal{H}_1$ over $\mathcal{H}_0$. Because we can compute Bayes factors from either direction, we must be careful to define our notation carefully. In this paper, I will adopt the common convention to define $B_{01}$ as the Bayes factor for $\mathcal{H}_0$ over $\mathcal{H}_1$. Similarly, $B_{10}$ would represent the Bayes factor for $\mathcal{H}_1$ over $\mathcal{H}_0$.  Note that, by our discussion above, $B_{01} = 1 / B_{10}$.

Though the previous discussion certainly speaks positively about the benefits of using the Bayes factor as a tool for inference, there are some important considerations that the researcher must address before implementing it as a tool for inference. First, as we'll see in the discussion below, the Bayes factor requires the analyst to specify prior distributions on all parameters in the underlying model. Thus, a given Bayes factor reflects a specific choice of prior. Second, the computation of these Bayes factors is usually nontrivial. For example, computing the either the numerator or denominator of Equation \ref{eq:odds} requires explicitly defining the hypothesis $\mathcal{H}_i$ as a model consisting of vectors $\bm{\xi}$ in some parameter space $\Xi$ and integrating the likelihood weighted by a prior distribution on these parameters; that is,
\[
  p(\bm{y}\mid \mathcal{H}_i) = \int_{\bm{\xi}\in \Xi} f(\bm{y}\mid \bm{\xi},\mathcal{H}_i)p(\bm{\xi} \mid \mathcal{H}_i)d\bm{\xi},
\]
where $f$ is the likelihood function, and $p$ denotes the prior distribution on parameters $\bm{\xi}\in \Xi$ under model $\mathcal{H}_i$. However, the last decade has seen the development of many tools intended to simplify the calculation of Bayes factors for the common models used by applied researchers, including online calculators, standalone software packages such as JASP \citep{JASP2018}, and a wide range of packages for the statistical computing environment R \citep{R}, including the package BayesFactor \citep{bayesfactor}. Because these solutions are designed to work across a variety of contexts, one must necessarily assume some defaults with respect to the models that underly these calculators. Many have argued that these defaults represent reasonable assumptions about the types of problems with which many applied researchers are concerned. However, recent advances in statistical computing have made it easier for the practicing researcher to build his or her own custom models for a given situation and compute Bayes factors to compare these models.

In this paper, I will provide a tutorial with particular focus on extending one type of Bayesian model comparison known as the Jeffreys-Zellner-Siow $t$-test \citep{rouder2009} (henceforth abbreviated as JZS $t$-test).  Specifically, I will describe a generalization that provides an adaptable, computationally efficient method for computing Bayes factors in a variety of single-sample and independent-samples designs. The organization of the paper is as follows. First, I will describe the mathematical underpinnings of the JZS $t$-test. Then, I will present two results which allow us to generalize the JZS $t$-test to a broader class of model comparisons: (1) the {\it Savage-Dickey density ratio} \citep{dickey1970,wagenmakers2010,wetzels2009}, which is used for comparing models in which one is a sharp hypothesis (e.g, a point null hypothesis) nested within another unconstrained model; and (2) the {\it encompassing prior technique} \citep{klugkist2005}, which is used for comparing nested models with ordinal constraints. Finally, I will demonstrate (with examples) how to use these techniques along with {\it posterior sampling} \citep{gelfand1990} to compute Bayes factors for model comparisons involving both point-null hypotheses (i.e., testing whether an effect is exactly 0), directional hypotheses (i.e., testing whether an effect is postive compared to whether it is negative), and interval-null hypotheses (i.e., testing whether an effect is approximately 0).

\section{The JZS $t$-test}

The JZS $t$-test \citep{rouder2009} was developed as a default Bayesian version of the orthodox $t$-test described above. We will denote by $\bm{y}$ a vector of observed data, and assume as above that $\bm{y}$ is normally distributed with mean $\mu$ and variance $\sigma^2$. We can then explicitly define two competing hypotheses: a null hypothesis $\mathcal{H}_0$ and an alternative hypothesis $\mathcal{H}_1$.  Both hypotheses can be parameterized with two parameters, $\mu$ and $\sigma^2$. Under the alternative hypothesis $\mathcal{H}_1$, we allow $\mu$ and $\sigma$ to freely vary. That is, $\mathcal{H}_1$ is an unconstrained model; $\mu$ could be positive, negative, or zero. For the null hypothesis $\mathcal{H}_0$, which states that the mean of data $\bm{y}$ is equal to 0, we can simply constrain $\mu$ to be 0. Thus, we say that $\mathcal{H}_0$ is {\it nested} within $\mathcal{H}_1$.  Under these models, we can compute the Bayes factor $B_{01} = p(\bm{y}\mid \mathcal{H}_0)/p(\bm{y}\mid \mathcal{H}_1)$, where
\[
  p(\bm{y}\mid \mathcal{H}_0) = \int_0^{\infty} f(\bm{y}\mid \mu=0, \sigma^2,\mathcal{H}_0)p(\sigma^2,\mathcal{H}_0)d\sigma^2
\]
and
\[
  p(\bm{y}\mid \mathcal{H}_1) = \int_{-\infty}^{\infty}\int_0^{\infty}f(\bm{y}\mid \mu,\sigma^2,\mathcal{H}_1)p(\mu,\sigma^2,\mathcal{H}_1)d\sigma^2 d\mu.
\]
These computations require placing priors on $\sigma^2$ under the null model $\mathcal{H}_0$ and both $\mu$ and $\sigma^2$ under the alternative model $\mathcal{H}_1$. Following \citet{jeffreys1961} and \citet{zellner1980}, Rouder et al. reparameterized the problem by placing a Cauchy prior on effect size $\delta=\mu/\sigma$. That is, under $\mathcal{H}_1$, $\delta \sim \text{Cauchy}(0,r)$, where $r$ represents the scale of expected effect sizes, and under $\mathcal{H}_0$, $\delta=0$. Rouder et al. placed a Jeffrey's prior on $\sigma^2$; specifically, $p(\sigma^2) \propto 1/\sigma^2$. With these default prior specifications, Rouder et al. (2009) showed that the Bayes factor can be computed as
%
\begin{equation}\label{eq:jzs}
  B_{01} = \frac{\left(1+\frac{t^2}{\nu}\right)^{-\frac{\nu+1}{2}}}{\int_0^{\infty}(1+Ngr^2)^{-\frac{1}{2}}\left(1+\frac{t^2}{(1+Ngr^2)\nu}\right)^{-\frac{\nu+1}{2}}(2\pi)^{-\frac{1}{2}}g^{-\frac{3}{2}}\exp\left(-\frac{1}{2g}\right)dg},
\end{equation}
%
where $t$ is the orthodox $t$ statistic, $r$ is the scale on the effect size prior, $N$ is the number of observations in $\bm{y}$, and $\nu=N-1$ denotes the degrees of freedom.

Though computationally convenient, this JZS Bayes factor formula has a few disadvantages. First, it reflects a very specific choice of prior specification. Though using a Cauchy prior on effect size may be a reasonable choice for many researchers, especially in the behavioral sciences \citep{rouder2009}, others may argue for a different prior. For example, \citet{killeen2007} used meta-analytic data from \citet{richard2003} to argued that effect sizes in social psychology are typically normally distributed with variance equal to 0.3. Certainly, one advantage of a Bayesian approach is that the prior on effect size should reflect the analyst's prior belief on what effect sizes should be expected. For some fields of study, these effect sizes may be expected to be small (i.e., social psychology), whereas in other fields, these effect sizes may be expected to be larger. Also, note that the reason for using the Cauchy prior (instead of a normal prior) in the JZS Bayes factor is one of computational convenience; it simply makes the computation work out. If the analyst wants to use a different prior, the formula for the JZS Bayes factor in Equation \ref{eq:jzs} would have to be recomputed, a task which would only be accessible to those researchers with the appropriate mathematical background.

Another disadvantage of the JZS Bayes factor is that it forces the analyst into a very specific hypothesis test; that is, $\mathcal{H}_0:\delta=0$ versus $\mathcal{H}_1:\delta\neq 0$. One may be interested instead in more flexible testing situations -- for example, testing a directional hypothesis $\mathcal{H}_0:\delta>0$ versus $\mathcal{H}_1:\delta\leq 0$, or an interval hypothesis such as $\mathcal{H}_0:-\varepsilon<\delta <\varepsilon$ versus an alternative model $\mathcal{H}_1:|\delta|>\varepsilon$ \citep{morey2011}. These tests would each require a major readjustment to the derivation of the JZS Bayes factor. Instead, I propose that we approach problems like these using a fundamentally different set of tools, which I will now describe.

\section{Generalizing the JZS $t$-test}

In this section, I describe a method for extending the default JZS $t$-test of Rouder and colleagues to use a wider class of priors on effect size. This generalization thus allows the analyst more freedom to specify a prior that may better reflect his or her a priori expectation about what effect sizes are typically encountered in a given field. The original description of this method is due to \citet{wetzels2009} and \citet{wagenmakers2010}, though the software implementation and specific extensions I will describe are novel.

The core method relies on a result known as the Savage-Dickey density ratio \citep{dickey1970}, which states that the Bayes factor for a pair of models in which one of the models is a one-point restriction of the other (i.e., $\mathcal{H}_0:\delta=0$ versus $\mathcal{H}_1:\delta \neq 0$) is simply the ratio of the ordinates of the one point of interest (i.e., $\delta=0$) in the posterior and prior densities, respectively. The technical formulation and proof will be presented momentarily. For now, however, one should appreciate that this can be a great simplification over other methods of computing Bayes factors presented above, since there is no need for integration. Assuming one can estimate the prior and posterior densities via some sampling method (e.g., Markov chain Monte Carlo, or MCMC, sampling), then the computation of this ratio of densities is straightforward.

The Savage-Dickey density ratio can stated rigorously as the following proposition:

\begin{prop} \label{savageDickey}(Savage-Dickey Density Ratio) Consider two competing models on data $\bm{y}$ containing parameters $\delta$ and $\varphi$, namely $\mathcal{H}_0:\delta=\delta_0,\varphi$ and $\mathcal{H}_1:\delta,\varphi$. In this context, we say that $\delta$ is a parameter of interest, $\varphi$ is a nuisance parameter (i.e., common to all models), and $\mathcal{H}_0$ is a sharp point hypothesis nested within $\mathcal{H}_1$. Suppose further that the prior for the nuisance parameter $\varphi$ in $\mathcal{H}_0$ is equal to the prior for $\varphi$ in $\mathcal{H}_1$ after conditioning on the restriction -- that is, $p(\varphi\mid \mathcal{H}_0) = p(\varphi\mid \delta=\delta_0,\mathcal{H}_1)$. Then
  \[
    B_{01} = \frac{p(\delta=\delta_0\mid \bm{y},\mathcal{H}_1)}{p(\delta=\delta_0\mid \mathcal{H}_1)}.
    \]
\end{prop}

\begin{proof}
By definition, the Bayes factor is the ratio of marginal likelihoods over $\mathcal{H}_0$ and $\mathcal{H}_1$, respectively. That is,
\begin{equation}\label{bf}
  B_{01}=\frac{p(\bm{y}\mid \mathcal{H}_0)}{p(\bm{y}\mid \mathcal{H}_1)}.
\end{equation}
The key idea in the proof is that we can use a ``change of variables'' technique to express $B_{01}$ entirely in terms of $\mathcal{H}_1$. This proceeds by first unpacking the marginal likelihood for $\mathcal{H}_0$ over the nuisance parameter $\varphi$ and then using the fact that $\mathcal{H}_0$ is a sharp hypothesis nested within $\mathcal{H}_1$ to rewrite everything in terms of $\mathcal{H}_1$. Specifically,
\begin{align*}
 p(\bm{y}\mid \mathcal{H}_0) &= \int p(\bm{y}\mid \varphi,\mathcal{H}_0)p(\varphi\mid \mathcal{H}_0)d\varphi\\
                              &= \int p(\bm{y}\mid \varphi,\delta=\delta_0,\mathcal{H}_1)p(\varphi\mid \delta=\delta_0,\mathcal{H}_1)d\varphi\\
  &=\ p(\bm{y}\mid \delta=\delta_0,\mathcal{H}_1).\\
\end{align*}
By Bayes' Theorem, we can rewrite this last line as
\[
  p(\bm{y}\mid \delta=\delta_0,\mathcal{H}_1) = \frac{p(\delta=\delta_0\mid \bm{y},\mathcal{H}_1)p(\bm{y}\mid \mathcal{H}_1)}{p(\delta=\delta_0\mid \mathcal{H}_1)}.
\]
Thus we have
\begin{align*}
  B_{01} = \frac{p(\bm{y}\mid \mathcal{H}_0)}{p(\bm{y}\mid \mathcal{H}_1)} &= p(\bm{y}\mid \mathcal{H}_0) \cdot \frac{1}{p(\bm{y}\mid \mathcal{H}_1)}\\[2mm]
  &= p(\bm{y}\mid \delta=\delta_0,\mathcal{H}_1) \cdot \frac{1}{p(\bm{y}\mid \mathcal{H}_1)}\\[2mm]
                                                                           &= \frac{p(\delta=\delta_0\mid \bm{y},\mathcal{H}_1)p(\bm{y}\mid \mathcal{H}_1)}{p(\delta=\delta_0\mid \mathcal{H}_1)} \cdot \frac{1}{p(\bm{y}\mid \mathcal{H}_1)}\\[2mm]
  &=\frac{p(\delta=\delta_0\mid \bm{y},\mathcal{H}_1)}{p(\delta=\delta_0\mid \mathcal{H}_1)}.
\end{align*}
\end{proof}

The beauty of Proposition \ref{savageDickey} is that it allows one to calculate the Bayes factor for a point null hypothesis (i.e, $\mathcal{H}_0:\delta=0$) by simply computing two densities: (1) the density of $\delta=0$ in the posterior, and (2) the density of $\delta=0$ in the prior. Then, the Bayes factor results by taking the ratio of these posterior and prior densities, respectively. Given this result, this changes the problem of computing Bayes factors from one of integration (e.g, the JZS Bayes factor) to one of estimating prior and posterior densities.

\subsection{Computing the Savage-Dickey density ratio}

The Savage-Dickey density ratio is an elegant solution to the problem of computing Bayes factors in situations involving a point null hypothesis $\mathcal{H}_0$. All that is required is that one can compute samples from the posterior of an effect size parameter under a specified alternative model $\mathcal{H}_1$. I think casting the problem in the context is preferable, not only for its flexibility, but especially given the broad class of computer methods now available for sampling posteriors in Bayesian models, including BUGS \citep{bugs}, JAGS \citep{jags}, and Stan \citep{stan}. I will now focus on one recent addition to this collection -- Greta \citep{greta}.

Greta is an R package designed for sampling from Bayesian models. It provides a reasonably simple language for modeling that is implemented directly within R, eliminating the need for writing models in another language (e.g., JAGS, Stan) and then having to call these external files from within R. Further, Greta uses the computational power of Google TensorFlow \citep{tensorflow}, so it provides fast convergence based on Hamiltonian Monte Carlo sampling \citep{neal2011}, it scales well to very large datasets, and it can even be configured to run on GPUs, providing the ability for massive parallel computation. Moreover, it is a free download from the Comprehensive R Archive Network (CRAN) \footnote{https://CRAN.R-project.org/package=greta}, and as such can be installed directly from within R by typing the command \verb|install.packages(``greta'')| at the R console. Note that fitting models with Greta will require the user to have a working installation of Python packages for TensorFlow (version 1.10.0 or higher) and tensorflow-probability (version 0.3.0 or higher). Once Greta is installed, the startup message will provide the user with system-specific instructions on how to install these two packages. While this step can be tricky, most errors can be addressed by following the recommendations on the Greta help page \footnote{\href{https://greta-stats.org/articles/get_started.html}{https://greta-stats.org/articles/get\_started.html}} and the TensorFlow help page \footnote{https://tensorflow.rstudio.com/tensorflow/articles/installation.html}.

To illustrate how Greta works, we will first look at a model inspired by that which was initially described by \citet{wetzels2009} as an alternative to the JZS $t$-test. As is common in these types of models, the model is depicted as a graphical model \citep{gilks1994} in Figure \ref{model1}.  In such graphical models, we use the various nodes to represent all variables of interest. Dependencies between these variables are indicated with graph structure. Deterministic nodes are denoted as rhombuses (i.e., rotated squares), whereas stochastic nodes are represented by unshaded circles. Finally, we denote observed variables by shaded nodes.

\begin{figure}
  \centering
  \includegraphics[width=0.7\textwidth]{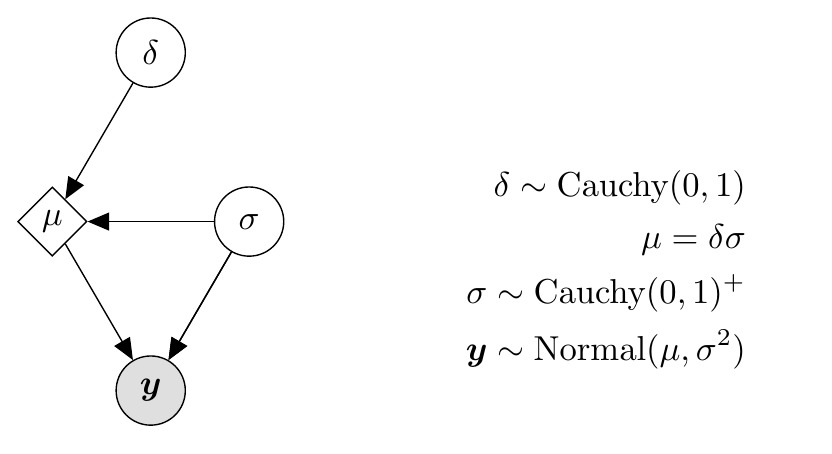}
  \caption{A graphical model for a posterior sampling Bayesian $t$-test. }
 \label{model1}
 \end{figure}

In this model, our data $\bm{y}$ is assumed to be drawn from a normal distribution with mean $\mu$ and variance $\sigma^2$. As in the discussion of the JZS $t$-test above, we consider effect size $\delta=\mu/\sigma$ as our main parameter of interest. For a fully Bayesian specification, we must place priors on $\sigma$ and $\delta$. For this first example, we follow \citet{rouder2009} and \citet{wetzels2009} and adopt their recommendations of placing a half-Cauchy prior on $\sigma$ (that is, one of the symmetric halves of the Cauchy(0,1) distribution that is defined for positive numbers only). Critically, we assume that effect size $\delta$ is distributed as Cauchy(0,1); the scale value of 1 indicates that, a priori, we believe that 50\% of our effect sizes would lie between -1 and 1.  Of course, as we'll see below, if this doesn't reflect the analyst's prior belief about $\delta$, this prior can be easily changed. This choice of model allows us to define two competing hypotheses:
\begin{align*}
  \mathcal{H}_0:&\hspace{2mm}\delta=0\\
  \mathcal{H}_1:&\hspace{2mm}\delta \sim \text{Cauchy}(0,1)
\end{align*}
As $\mathcal{H}_0$ is a sharp hypothesis nested within $\mathcal{H}_1$, we can apply Proposition \ref{savageDickey} (the Savage-Dickey density ratio) and compute
\[
  B_{01} = \frac{p(\delta=0\mid \bm{y},\mathcal{H}_1)}{p(\delta=0\mid \mathcal{H}_1)}.
\]
To do this, we'll need to draw samples from the posterior distribution of $\delta$ under $\mathcal{H}_1$ and estimate the height of $\delta=0$ in an estimated density function for this posterior. All of this can be done in R, as I will now illustrate.

To begin, let us consider a simple example, the type of which can be found in most elementary statistics textbooks. This example comes from \citet{hoel1984}. Suppose that 10 patients take part in an experiment on a new drug that is supposed to increase sleep in the patients. Table \ref{tab:ex1} shows the hours of sleep gained by each patient (negative values indicate lost sleep). Assuming that the sample data are normally distributed with mean $\mu$ and variance $\sigma^2$, we can test the hypothesis $\mathcal{H}_0:\mu=0$ against $\mathcal{H}_1:\mu \neq 0$.
%
\begin{table}
  \setlength{\tabcolsep}{10pt}
  \centering
  \caption{Example data for a single-sample $t$ test}
  \label{tab:ex1}
  \begin{tabular}{|l|cccccccccc|}
    \hline
    {\it Patient} & 1 & 2 & 3 & 4 & 5 & 6 & 7 & 8 & 9 & 10\\
    {\it Hours gained} & 0.7 & -1.1 & -0.2 & 1.2 & 0.1 & 3.4 & 3.7 & 0.8 & 1.8 & 2.0\\
    \hline
  \end{tabular}
\end{table}
%
The R code necessary to perform a Bayesian $t$-test on this data is displayed in Listing \ref{code1}. The first step will be to load the Greta library (see line 2). After this, we need to assign our sample data to a vector $\bm{y}$ and then convert these to $z$-scores (see lines 5-6). The next step is to define the prior distributions on $\delta$ and $\sigma$. The Greta syntax allows this to be done in a quite straightforward manner (see lines 9-10). Further, any deterministic operations should then be defined, as we do in line 13. Then, we can define our likelihood for the $z$-scores. The wording of the syntax has a nice advantage here, as it describes exactly what we are assuming about our scores; namely, that they are normally distributed with mean $\mu$ and variance $\sigma^2$ (see line 16). The last step in setting up the model is to {\it define} the model; that is, we collect all of the variables of interest in our analysis. We have three: $\mu$, $\sigma$, and $\delta$, which we collect together in line 19. Now we are finally ready to sample from the posterior distributions of the variables in our model. We will focus our interest on \verb|delta|, but Greta will automatically sample all posteriors for us. This step, displayed in line 22, will take a little while, depending on computing resources.

\begin{lstlisting}[float,caption=Building and sampling from the single-sample $t$-test model,label=code1]
# load libraries
library(greta)

# data from Hoel (1984) 
y = c(0.7,-1.1,-0.2,1.2,0.1,3.4,3.7,0.8,1.8,2.0)
z = y/sd(y)

# priors
delta = cauchy(0,1)
sigma = cauchy(0,1,truncation = c(0,Inf))

# operations
mu = delta*sigma

# likelihood
distribution(z) = normal(mu,sigma)

# define model
m = model(mu, sigma, delta)

# draw samples
draws = mcmc(m, n_samples=5000)
\end{lstlisting}

Once the sampling is complete, there are two ways to inspect the samples before proceeding to our inference. The first is to type \verb|summary(draws)|; this will show us various descriptive statistics of the samples, including mean, standard deviation, standard error, and quantiles. In our example, there will be three lines of output; one for each of $\mu$, $\sigma$, and $\delta$. Another way to look at the samples is to inspect the path of the samples over time as they explore the posterior distributions.  This is done by using the \verb|mcmc_trace| command from the bayesplot package \footnote{https://CRAN.R-project.org/package=bayesplot} \citep{bayesplot}. If the samples converged appropriately, one should see the characteristic ``hairy caterpillar'' plot, indicating that the chains mixed well and truly randomly explored the posteriors.

We will now look at how to compute Bayes factors necessary to compare the models $\mathcal{H}_0$ and $\mathcal{H}_1$. We will do this by computing the Savage-Dickey density ratio. Recall from Proposition \ref{savageDickey} that in order to compute $B_{01}$, we simply need to compute the ordinate of $\delta=0$ in the densities of the prior and posterior, and then take their ratio. We already know the density function for the prior, and it is implemented in R as the \verb|dcauchy| function. However, since we are using samples to approximate the posterior, we need a way to estimate its density function from the samples. One such method is to use a logspline density estimator \citep{kooperberg1992,stone1997}, which is implemented by the function \verb|logspline| from the R package polspline \footnote{https://CRAN.R-project.org/package=polspline} \citep{polspline}.

Listing \ref{code2} shows the R code necessary to both (1) plot the ordinates from the prior and posterior densities for $\delta=0$ under $\mathcal{H}_1$, and (2) compute $BF_{01}$ as the ratio of these ordinates. The first step is to extract the relevant samples from the posterior for $\delta$ from the object \verb|draws| (see line 4). Note that if the analyst is interested in other parameters (e.g., $\mu$ or $\sigma$, the \verb|[,3]| part of line 4 can be adjusted appropriately. We can then perform the logspline estimate of the posterior density for $\delta$, as shown on line 5.

From here, there are two paths worth exploring. First, I typically will produce a plot showing the components of the Savage-Dickey density ratio. Such a plot will consist of (1) a plot of the prior density for $\delta$ under $\mathcal{H}_1$ (the Cauchy(0,1) distribution); (2) a plot of the posterior density (from the logspline estimate); and (3) the ordinates of $\delta=0$ in both of these densities. Lines 7-21 will produce such a graph, which can be seen in Figure \ref{fig:ex1}.

Next, we can compute the Savage-Dickey density ratio using the code on lines 24-30. Lines 24-25 compute the specific ordinates required. \verb|posterior| represents the ordinate of $\delta=0$ in the posterior distribution.  Since the posterior was estimated from the \verb|logspline| function, we call this estimate for our calculation using \verb|dlogspline| along with the object name of our estimate from line 5 (i.e, \verb|fitPost|). \verb|prior| represents the ordinate of $\delta=0$ under $\mathcal{H}_1$; computing this uses a simple call to the \verb|dcauchy| function. Finally, we can divide \verb|posterior| by \verb|prior| to compute $B_{01}$, which is denoted in line 26 by \verb|BF01|.  Lines 29 and 30 then simply display both $B_{01}$ (the Bayes factor in favor of $\mathcal{H}_0$ over $\mathcal{H}_1$) and $B_{10}$ (the Bayes factor in favor of $\mathcal{H}_1$ over $\mathcal{H}_0$). From this computation, one can see that we have moderate support for $\mathcal{H}_1$, as $B_{01} \approx 0.4$. The intuition for this can be had by looking at how density at $\delta=0$ changes from prior to posterior. In Figure \ref{fig:ex1}, the posterior density of $\delta=0$ decreases relative to the prior density, indicating that our belief in a null effect decreases after observing data $\bm{x}$. Equivalently, we can use the reciprocal to compute $B_{10} \approx 2.6$, indicating that our data are 2.6 times more likely under the alternative model $\mathcal{H}_1$ compared to the null model $\mathcal{H}_0$, giving us positive evidence in favor of a nonzero effect $\delta$.

\begin{lstlisting}[float,caption=Plotting and computing the Savage-Dickey density ratio,label=code2]
Library(polspline)
  
# extract draws from MCMC object and fit a density estimate
posteriorDelta = draws[[1]][,3]
fitPost = logspline(posteriorDelta)

x = seq(-2,2,length.out=1000)

# plot density of prior and posterior together
plot(x, dlogspline(x, fitPost), type="l", main="", xlab="delta",
     ylab="density", xlim=c(-2,2), lwd=2, lty=1)

# add prior
lines(x, dcauchy(x,0,1), lwd=2, lty=3)

# add points at 0 for both prior and posterior
points(0, dlogspline(0, fitPost), pch=19)
points(0, dcauchy(0,0,1), pch=19)

legend(-2,0.8, legend=c("Prior density","Posterior density"),
     lty=c(3,1), lwd=c(2,2), bty="n")

# compute SD density ratio
posterior <- dlogspline(0, fitPost) 
prior     <- dcauchy(0,0,1)                
BF01      <- posterior/prior

# display both Bayes factors
BF01
1/BF01    
\end{lstlisting}

\begin{figure}
  \centering
  \includegraphics[width=\linewidth]{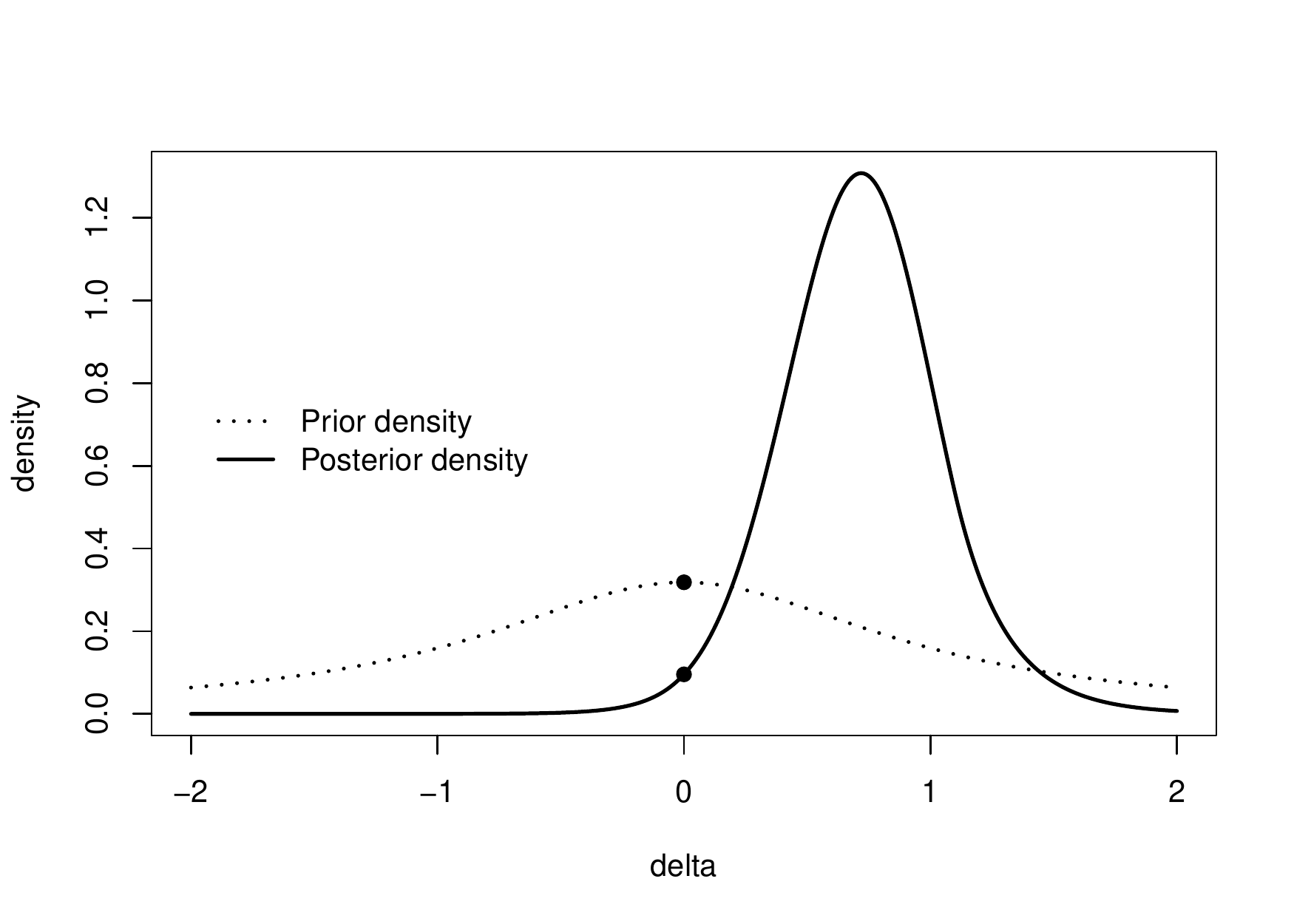}
    \caption{A plot showing the necessary components for computing the Savage-Dickey density ratio. Included are the density of the prior for $\delta$ under $\mathcal{H}_1$ (i.e., a Cauchy(0,1) distribution, depicted as a dashed line) as well as the logspline density estimate for the posterior of $\delta$ (depicted as a solid line). The Bayes factor $B_{01}$ can be computed as the ratio of the ordinates of $\delta=0$ under the posterior and prior, respectively.}
     \label{fig:ex1}
  \end{figure}

  Now, suppose the analyst had a different a prior belief about the effect sizes he or she would expect in this context. For illustration, let us suppose that $\delta$ is normally distributed with mean $\mu=0$ and variance $\sigma^2=0.3$, as recommended by \citet{killeen2007}. In this case, all of the above code could be run again with some minor changes. First, one would need to change line 9 in Listing 1 to \verb|delta = normal(0,sqrt(0.3))|. Also, any line in Listing 2 that uses \verb|dcauchy(x,0,1)| would need to be replaced by \verb|dnorm(x,0,sqrt(0.3))|. After this minor change, the resulting Bayes factor would be $B_{01}\approx 0.25$, or equivalently, $B_{10}\approx 4.0$. Note that this Bayes factor is a bit larger than the previous one in which the Cauchy prior was used for $\delta$. The reason folllows from the fact that the Cauchy prior is more dispersed relative to a normal prior, and thus with a normal prior, we have a relatively greater prior mass on smaller effects. Particularly, the ordinate of $\delta=0$ is larger in the normal prior compared to the Cauchy prior. The result is that the ratio between the ordinates of $\delta=0$ in the prior and posterior becomes larger for the normal prior, thus giving us a larger Bayes factor.

\subsection{Using the Savage-Dickey density ratio for a two-sample design}

The methods described above will readily scale up to problems involving two independent samples. All that is required is that the underlying model is adjusted accordingly. I will illustrate this with another example from \citet{hoel1984}.

Consider a sample of 20 rats, each of which receives their main source of protein from either raw peanuts or roasted peanuts. To compare weight gains as a function of protein source, a researcher randomly assigns 10 rats to receive only raw peanuts and 10 rats to receive only roasted peanuts. The resulting weight gains (in grams) are displayed in Table \ref{tab:rats}.

\begin{table}
  \centering
    \setlength{\tabcolsep}{10pt}
  \caption{Example data for a two-sample $t$-test.}
  \label{tab:rats}
  \begin{tabular}{|l|cccccccccc|}
    \hline
    {\it Raw} & 62 & 60 & 56 & 63 & 56 & 63 & 59 & 56 & 44 & 61\\
    {\it Roasted} & 57 & 56 & 49 & 61 & 55 & 61 & 57 & 54 & 62 & 58\\
    \hline 
  \end{tabular}
\end{table}

First, we must consider the underlying model. As with the single-sample example, we can represent this model as a directed acyclic graph, which is shown in Figure \ref{fig:model2}. In this model \citep[inspired by][]{wetzels2009}, both independent samples $\bm{x}$ and $\bm{y}$ are assumed to be drawn from two normal distributions with shared variance $\sigma^2$. The mean of the parent distribution of $\bm{x}$ is $\mu+\alpha/2$ and the mean for the parent distribution of $\bm{y}$ is $\mu-\alpha/2$. With this parameterization, $\alpha$ represents the ``effect'' or difference between the two populations. As with the single-sample example, we then scale this effect to a standardized effect $\delta=\alpha/\sigma$. Also, standard Cauchy priors are placed on $\delta$, $\mu$, and a truncated Cauchy prior is placed on $\sigma$.

\begin{figure}
  \centering
  \includegraphics[width=0.7\textwidth]{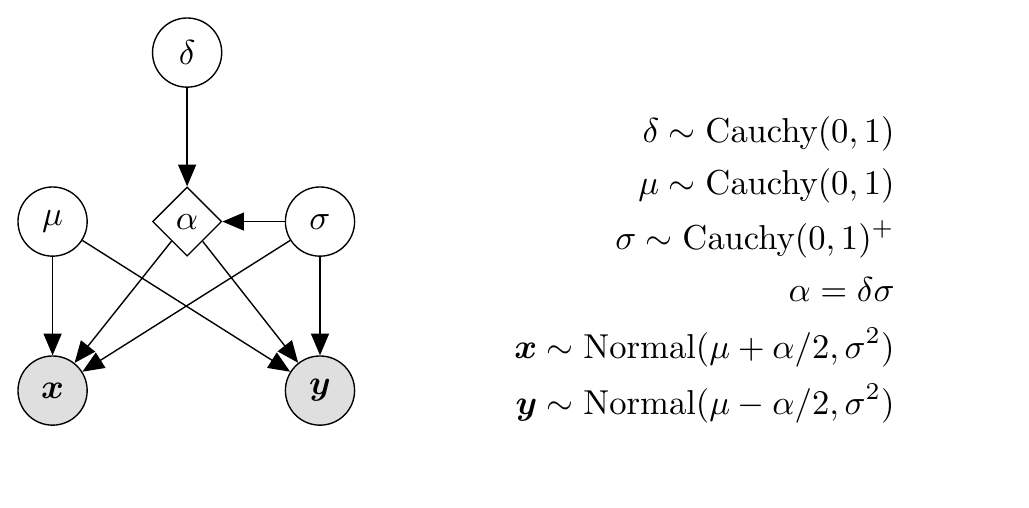}
  \caption{A graphical model for the posterior sampling independent samples $t$-test.}
  \label{fig:model2}
\end{figure}

For concreteness, let us denote the sample of weight gains from the raw peanut diet as $\bm{x}$ and the weight gains from roasted peanuts as $\bm{y}$. Given the model in Figure \ref{fig:model2}, our goal is to sample from the posterior distribution of $\delta$. The R code necessary to perform this sampling is displayed in Listing \ref{code4}. The procedure is similar to the single-sample model in Listing \ref{code1}, but there are some notable modifications that are particular to the independent samples model. First, we need to rescale the raw data vectors \verb|x| and \verb|y| to $z$-scores. Since we are assuming shared variance, it suffices to base both $z$-score transformations on only one of \verb|x| and \verb|y|. In lines 9-10, I have chosen to base the $z$-scores on \verb|x|, but note that similar results will be obtained if instead the researcher chooses to base all $z$-scores on \verb|y|. Lines 13-15 define the priors that we assigned to the parameters in our model. Lines 18-24 reflect our assumption that data $\bm{x}$ and $\bm{y}$ are randomly drawn from two normal distributions centered at $\mu+\alpha/2$ and $\mu-\alpha/2$, respectively. In lines 27 and 30 we tell Greta to pull 5000 samples from the posterior distribution of $\delta$; note that for simplicity, I have only included \verb|delta| in the model, though one could add any other variable in the model if desired. These posterior samples are then extracted into the vector \verb|posteriorDelta| in line 33. 

\begin{lstlisting}[float,caption=Building and sampling from the independent-samples $t$-test model,label=code4]
# load libraries
library(greta)

# data from Hoel (1984)
x = c(62, 60, 56, 63, 56, 63, 59, 56, 44, 61)
y = c(57, 56, 49, 61, 55, 61, 57, 54, 62, 58)

# rescale so that x has mean=0 and sd=1
zx = (x-mean(x))/sd(x)
zy = (y-mean(x))/sd(x)

# priors
delta = cauchy(0,1)
mu = cauchy(0,1)
sigma = cauchy(0,1,truncation = c(0,Inf))

# operations
alpha = delta*sigma
mux = mu + alpha/2
muy = mu - alpha/2

# likelihood
distribution(zx) = normal(mux,sigma)
distribution(zy) = normal(muy,sigma)

# define model
m = model(delta)

# MCMC sample
draws = mcmc(m, n_samples = 5000)

# extract draws from MCMC object
posteriorDelta = draws[[1]][,1]
\end{lstlisting}

After completing the code in Listing \ref{code4}, the Savage-Dickey density ratio can be plotted and computed as we did earlier in Listing \ref{code2}. Figure \ref{fig:ex2} shows this ratio graphically; indeed, note that the posterior density of $\delta=0$ increases relative to the prior density. This ratio is computed to be $B_{01}=2.92$, indicating that the data are 2.92 times more likely under $\mathcal{H}_0$ than under $\mathcal{H}_1$. Thus, we can conclude positive support for a null effect of peanut type on rats' weight gain.

\begin{figure}
  \centering
  \includegraphics[width=\linewidth]{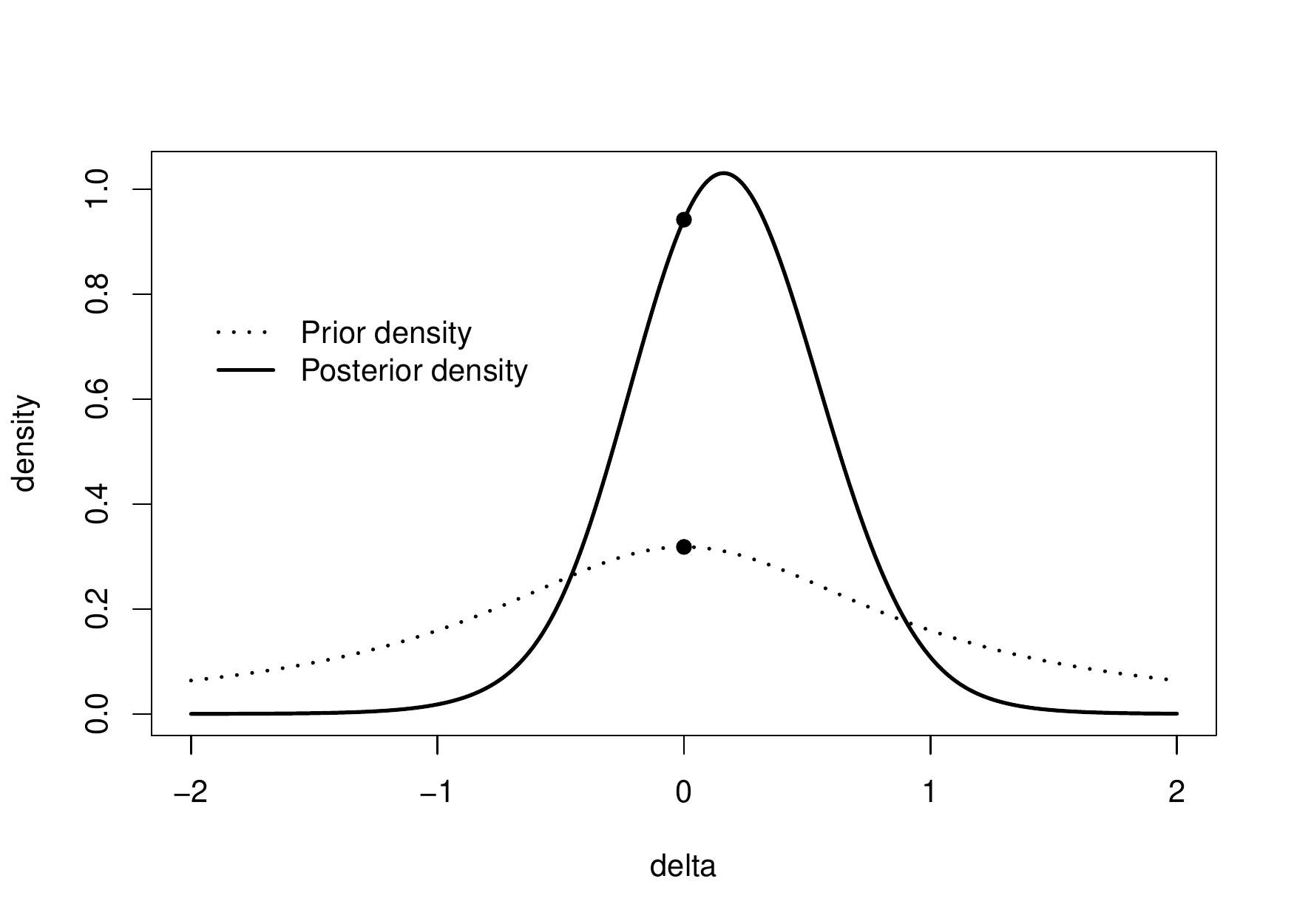}
    \caption{A plot showing the necessary components for computing the Savage-Dickey density ratio in the independent-samples model. Included are the density of the prior for $\delta$ under $\mathcal{H}_1$ (i.e., a Cauchy(0,1) distribution, depicted as a dashed line) as well as the logspline density estimate for the posterior of $\delta$ (depicted as a solid line). The Bayes factor $B_{01}$ can be computed as the ratio of the ordinates of $\delta=0$ under the posterior and prior, respectively.}
     \label{fig:ex2}
  \end{figure}

\section{Extension: using encompassing priors for inequality constraints}

In the previous section, I described an extension of the JZS $t$-test that uses MCMC sampling to approximate the posterior distribution of effect size $\delta$. This method works for sharp hypotheses (i.e., a point null, such as $\mathcal{H}_0:\delta=0$) by employing the Savage-Dickey density ratio, which reduces the calculation of the Bayes factor $B_{01}$ into a simple ratio based on the ordinates of the point $\delta=0$ in both the prior and posterior distributions for $\delta$.

Consider again the sleep example above. What if instead the researcher wanted to whether the new drug {\it increased} sleep in patients? This would require the ability to test a directional hypotheses $\mathcal{H}_1:\delta > 0$ against $\mathcal{H}_0:\delta\leq 0$. At first glance, this seems like quite a different problem, as the Savage-Dickey density ratio does not directly apply to models with inequality constraints. However, there is a method due originally to \citet{klugkist2005} that fits with this type of problem. In their approach, Klugkist et al. cast such problems as one of testing models with inequality constraints nested within an encompassing model. In this context, both hypotheses $\mathcal{H}_0$ and $\mathcal{H}_1$ are considered as specific inequality constraints nested within an encompassing model $\mathcal{H}_e:\delta$, where $\delta$ is unconstrained (i.e., $\delta \in \mathbb{R}$). The Klugkist et al. approach (which I will hereafter call the {\it encompassing approach}) amounts to using MCMC samples to calculate
\[
  B_{0e} = \frac{p(\bm{y}\mid \mathcal{H}_0)}{p(\bm{y}\mid \mathcal{H}_e)}
\]
and
\[
  B_{1e} = \frac{p(\bm{y}\mid \mathcal{H}_1)}{p(\bm{y}\mid \mathcal{H}_e)}.
\]
Once these two Bayes factors are computed, one can use transitivity of Bayes factors to compute
\[
  B_{01} = B_{0e}\cdot B_{e1} = B_{0e} \cdot \frac{1}{B_{1e}}.
\]

The mechanics of the encompassing approach can be summarized in the following proposition:

\begin{prop}\label{encompassing} Consider two models $\mathcal{H}_1$ and $\mathcal{H}_e$, where $\mathcal{H}_1$ is nested within an encompassing model $\mathcal{H}_e$ via an inequality constraint on some parameter $\delta$, and $\delta$ is unconstrained under $\mathcal{H}_e$. Then
\[
  B_{1e} = \frac{c}{d} = \frac{1/d}{1/c}
\]
where $1/d$ and $1/c$ represent the proportions of the posterior and prior of the encompassing model, respectively, that are in agreement with the inequality constraint imposed by the nested model $\mathcal{H}_1$.
\end{prop}

\begin{proof}
Consider first that for any model $\mathcal{H}_t$ on data $\bm{y}$ with parameter vector $\bm{\xi}$, Bayes' theorem implies
\[
  p(\bm{\xi}\mid \bm{y},\mathcal{H}_t) = \frac{f(\bm{y}\mid \bm{\xi},\mathcal{H}_t)\cdot p(\bm{\xi}\mid \mathcal{H}_t)}{p(\bm{y}\mid \mathcal{H}_t)}.
\]
Thus, we can write the marginal likelihood for $\bm{y}$ under $\mathcal{H}_t$ as
\[
  p(\bm{y}\mid \mathcal{H}_t) = \frac{f(\bm{y}\mid \bm{\xi},\mathcal{H}_t)\cdot p(\bm{\xi}\mid \mathcal{H}_t)}{p(\bm{\xi}\mid \bm{y},\mathcal{H}_t)}.
\]
Taking the ratio of the marginal likelihoods for $\mathcal{H}_1$ and the encompassing model $\mathcal{H}_e$ yields the following Bayes factor:
\[
  B_{1e} = \frac{f(\bm{y}\mid \bm{\xi},\mathcal{H}_1)\cdot p(\bm{\xi}\mid \mathcal{H}_1) / p(\bm{\xi}\mid \bm{y},\mathcal{H}_1)}{f(\bm{y}\mid \bm{\xi},\mathcal{H}_e)\cdot p(\bm{\xi}\mid \mathcal{H}_e) / p(\bm{\xi}\mid \bm{y},\mathcal{H}_e)}.
\]
Now, both the constrained model $\mathcal{H}_1$ and the encompassing model $\mathcal{H}_e$ contain the same parameters $\bm{\xi}$. Choose a specific value of $\bm{\xi}$, say $\bm{\xi'}$, that exists in both models $\mathcal{H}_1$ and $\mathcal{H}_e$ (we can do this because $\mathcal{H}_1$ is nested within $\mathcal{H}_e$. Then, for this parameter value $\bm{\xi'}$, we have $f(\bm{y}\mid \bm{\xi'},\mathcal{H}_1)=f(\bm{y}\mid \bm{\xi'},\mathcal{H}_2)$, so the expression for the Bayes factor reduces to an expression involving only the priors and posteriors for $\bm{\xi'}$ under $\mathcal{H}_1$ and $\mathcal{H}_e$:
\[
  B_{1e} = \frac{p(\bm{\xi'}\mid \mathcal{H}_1) / p(\bm{\xi'}\mid \bm{y},\mathcal{H}_1)}{p(\bm{\xi'}\mid \mathcal{H}_e) / p(\bm{\xi'}\mid \bm{y},\mathcal{H}_e)}.
\]
Because $\mathcal{H}_1$ is nested within $\mathcal{H}_e$ via an inequality constraint, the prior $p(\bm{\xi'}\mid \mathcal{H}_1)$ is simply a truncation of the encompassing prior $p(\bm{\xi'}\mid \mathcal{H}_e)$. Thus, we can express $p(\bm{\xi'}\mid \mathcal{H}_1)$ in terms of the encompassing prior $p(\bm{\xi'}\mid \mathcal{H}_e)$ by multiplying the encompassing prior by an indicator function over $\mathcal{H}_1$ and then normalizing the resulting product.  That is,
\begin{align*}
  p(\bm{\xi'}\mid \mathcal{H}_1) &= \frac{p(\bm{\xi'}\mid \mathcal{H}_e)\cdot I_{\bm{\xi'}\in\mathcal{H}_1}}{\int p(\bm{\xi'}\mid\mathcal{H}_e)\cdot I_{\bm{\xi'}\in \mathcal{H}_1}d\bm{\xi'}}\\
                                 &=\Biggl(\frac{I_{\bm{\xi'}\in \mathcal{H}_1}}{\int p(\bm{\xi'}\mid \mathcal{H}_e)\cdot I_{\bm{\xi'}\in\mathcal{H}_1}d\bm{\xi'}}\Biggr)\cdot p(\bm{\xi'}\mid \mathcal{H}_e),
\end{align*}
where $I_{\bm{\xi'}\in \mathcal{H}_1}$ is an indicator function. For parameters $\bm{\xi'}\in \mathcal{H}_1$, this indicator function is identically equal to 1, so the expression in parentheses reduces to a constant, say $c$, allowing us to write
\[
  p(\bm{\xi'}\mid \mathcal{H}_1) = c\cdot p(\bm{\xi'}\mid \mathcal{H}_e).
\]
By similar reasoning, we can write the posterior as
\[
  p(\bm{\xi'}\mid \bm{y},\mathcal{H}_1) = \Biggl(\frac{I_{\bm{\xi'}\in \mathcal{H}_1}}{\int p(\bm{\xi'}\mid \bm{y},\mathcal{H}_e)I_{\bm{\xi'}\in\mathcal{H}_1}d\bm{\xi'}}\Biggr)\cdot p(\bm{\xi'}\mid \bm{y},\mathcal{H}_e) = d\cdot p(\bm{\xi'}\mid \bm{y},\mathcal{H}_e).
\]
This gives us
\[
  B_{1e} = \frac{c\cdot p(\bm{\xi'}\mid \mathcal{H}_e) / d\cdot p(\bm{\xi'}\mid \bm{y},\mathcal{H}_e)}{p(\bm{\xi'}\mid \mathcal{H}_e) / p(\bm{\xi'}\mid \bm{y},\mathcal{H}_e)} = \frac{c}{d} = \frac{1/d}{1/c}.
\]
Note that by definition, $1/d$ represents the proportion of the posterior distribution for $\bm{\xi}$ under the encompassing model $\mathcal{H}_e$ that agrees with the constraints imposed by $\mathcal{H}_1$.  Similarly, $1/c$ represents the proportion of the prior distribution for $\bm{\xi}$ under the encompassing model $\mathcal{H}_e$ that agrees with the constraints imposed by $\mathcal{H}_1$.
\end{proof}

It might seem a bit odd to represent the fraction $c/d$ in the form $(1/d)/(1/c)$. However, this is again done for a computational advantagem, as we can use MCMC sampling to easily estimate the proportions $1/d$ and $1/c$. Also note that in some sense, the encompassing prior approach of \citet{klugkist2005} is a generalized version of the Savage-Dickey density ratio. Indeed, \citet{wetzels2010} proved that under ``about equality'' constraints (e.g., a constrained model $\mathcal{M}:-\varepsilon<\delta < \varepsilon$ for $\varepsilon>0$), the Bayes factor derived from the encompassing approach tends toward the Bayes factor (for the point null where $\delta=0$) obtained from the Savage-Dickey density ratio as $\varepsilon\rightarrow 0$.

\subsection{Computing Bayes factors with the encompassing approach}

To illustrate the computation of Bayes factors with the encompassing approach, let us consider the problem mentioned immediately above -- suppose we wanted to test whether the drug that we administered to sleep patients actually {\it increased} the patients' sleep. Specifically, we wish to compare $\mathcal{H}_0:\delta \leq 0$ against $\mathcal{H}_1:\delta >0$. We will do this by considering both $\mathcal{H}_0$ and $\mathcal{H}_1$ as models with inequality constraints nested with an encompassing model $\mathcal{H}_e:\delta$, where $\delta$ is unconstrained. Then, we can use transitivity to compute $B_{10} = B_{1e}\cdot B_{e0} = B_{1e}/B_{0e}$.

The R code necessary to perform these computations is in Listing \ref{code3}. As the encompassing model is defined identical to that from Figure \ref{model1}, the code assumes that we have already drawn samples from that model, as we did in Listing \ref{code1}. Note that just like our previous computations with the Savage-Dickey density ratio, using the encompassing approach requires that we sample from the posterior of $\delta$ under the unconstrained, encompassing model $\mathcal{H}_e$. Though the notation is different, this is exactly the same posterior distribution that we sampled from in Listing \ref{code1}.

\begin{lstlisting}[float,caption=Computing a Bayes factor for a directional hypothesis using the encompassing approach,label=code3]
# directional hypothesis using encompassing priors
# H_0: "null" model: delta <= 0
# H_1: "alternative" model: delta > 0
# H_e: "encompassing" model: delta ~ Cauchy (unconstrained)

# H_0 versus H_e
postEvidential = mean(posteriorDelta<=0)
priorEvidential = pcauchy(0,0,1)
BF0e = postEvidential/priorEvidential

# H_1 versus H_e
postEvidential = mean(posteriorDelta>0)
priorEvidential = 1-pcauchy(0,0,1)
BF1e = postEvidential/priorEvidential

# H_1 versus H_0
BF10 = BF1e/BF0e; BF10
\end{lstlisting}

The key steps in Listing \ref{code3} are as follows. First, we will compare $\mathcal{H}_0$ to the encompassing model $\mathcal{H}_e$. To this end, we need to compute the proportion of posterior samples from the encompassing model that are in agreement with the inequality constraint imposed by $\mathcal{H}_0$ (this is the quantity $1/d$ in the proof of Proposition \ref{encompassing}). We say that such samples are ``evidential'' of $\mathcal{H}_0$. The R code that will compute this proportion is in line 7. Then, we need to compute the proportion of evidential samples in the prior (i.e., $1/c$). Since the prior has known density $\delta \sim \text{Cauchy}(0,1)$, we can use the \verb|pcauchy| command to directly compute the proportion of values $\delta$ that are less than 0; this computation proceeds in line 8. Then, by Proposition \ref{encompassing}, we can simply divide these two quantities to compute $B_{0e}$ (see line 9).

Next, we do a similar computation with $\mathcal{H}_1$ versus the encompassing model $\mathcal{H}_e$, shown in lines 12-14. This gives us a value for $B_{1e}$. Now, we can compute the Bayes factor for $\mathcal{H}_1$ over $\mathcal{H}_0$ by computing $B_{1e}/B_{0e} \approx 65$, indicating that the observed data are approximately 65 times more likely under the alternative model $\mathcal{H}_1:\delta>0$ compared to the null model $\mathcal{H}_0:\delta\leq 0$.

This approach can be extended to test a wide variety of hypotheses involving inequality constraints. One particular advantage of the encompassing approach is that it gives us the ability to test interval null hypotheses -- that is, hypotheses of the form $\mathcal{H}_0:-\varepsilon < \delta < \varepsilon$. To illustrate, consider the analyst who is not interested in whether an effect is {\it exactly} 0, but rather, is interested in whether an effect is larger than threshold, say $\varepsilon = 0.2$.

An example of such computation is displayed in Listing \ref{code4}. Like in the example above, we define three hypotheses: two competing hypotheses $\mathcal{H}_0:|\delta|<\varepsilon$ and $\mathcal{H}_1:|\delta|>\varepsilon$, both nested within an encompassing model $\mathcal{H}_e:\delta$. In the example, I have set $\varepsilon=0.20$, but one can set this value at whatever value seems reasonable for the given context.

\begin{lstlisting}[float,caption=Computing a Bayes factor for an interval null hypothesis using the encompassing approach,label=code4]
# interval null via encompassing priors 
# H_0: "null" model: |delta| < epsilon
# H_1: |delta| > epsilon
# H_e: delta ~ Cauchy (unconstrained)

epsilon = 0.2

# H_0 versus H_e
postEvidential = mean(abs(posteriorDelta) < epsilon)
priorEvidential = pcauchy(epsilon,0,1)-pcauchy(-epsilon,0,1)
B_0e = postEvidential/priorEvidential

# H_1 versus H_e
postEvidential = mean(abs(posteriorDelta) > epsilon)
priorEvidential = (1-pcauchy(epsilon,0,1)) + pcauchy(-epsilon,0,1)
B_1e = postEvidential/priorEvidential

# H_1 versus H_0
B_10 = B_1e/B_0e; B_10
\end{lstlisting}

As in the example before, we use the posterior samples for $\delta$ under $\mathcal{H}_e$ that were generated in Listing \ref{code1} to calculate the proportions of the posterior that satisfied the inequality constraints on $\delta$ imposed by $\mathcal{H}_0$ and $\mathcal{H}_1$. These computations are performed in lines 9 and 14. The relevant proportions of the unconstrained prior that obey the imposed inequality constraints are again calculated using the \verb|pcauchy| command, as seen in lines 10 and 15. Finally, lines 11,16, and 19 calculate the relevant Bayes factors. As we can see, the Bayes factor $B_{10} \approx 2.2$, indicating that the observed data are 2.2 times more likely under the model $\mathcal{H}_1:|\delta|>\varepsilon$ compared to the model $\mathcal{H}_0:|\delta|<\varepsilon$. Notice that this is similar to, but less than, the Bayes factor obtained with the point null hypothesis $\mathcal{H}_0:\delta=0$ from earlier in the paper.

\section{Conclusions}
In this tutorial, I have demonstrated a flexible approach to extending the default JZS $t$-test, a Bayesian test that is becoming increasingly popular in the social and behavioral sciences \citep{rouder2009}. The approach uses two theoretical results, the Savage-Dickey density ratio \citep{dickey1970} and the method of encompassing priors \citep{klugkist2005} in combination with an easy-to-use probabilistic modeling package for R called Greta \citep{greta}. Though the examples presented in this paper are quite trivial to implement, they provide the reader with a general workflow that can be extended to solve problems relevant to his or her own work. Inherent in the techniques presented here is flexibility; the user has complete freedom to specify the underlying models and specific model comparisons in any way that he or she wishes. Finally, the Greta modeling language is easy to learn and readily extends to more complex modelsy. Furthermore, by harnessing the power of Google Tensorflow \citep{tensorflow}, the MCMC sampler is fast, with all models described in the paper converging in less than one minute. In summary, I think this is an advantageous approach to using default Bayesian tests for common hypothesis testing scenarios, especially those common in the social, behavioral, and other applied sciences.

\section*{Acknowledgement}

I am grateful to the handling editor and two anonymous referees for their comments on an earlier version of this manuscript.

\bibliographystyle{apalike}
\bibliography{references}

\begin{thebibliography}{}

\bibitem[Abadi et~al., 2015]{tensorflow}
Abadi, M., Agarwal, A., Barham, P., Brevdo, E., Chen, Z., Citro, C., Corrado,
  G.~S., Davis, A., Dean, J., Devin, M., Ghemawat, S., Goodfellow, I., Harp,
  A., Irving, G., Isard, M., Jia, Y., Jozefowicz, R., Kaiser, L., Kudlur, M.,
  Levenberg, J., Man\'{e}, D., Monga, R., Moore, S., Murray, D., Olah, C.,
  Schuster, M., Shlens, J., Steiner, B., Sutskever, I., Talwar, K., Tucker, P.,
  Vanhoucke, V., Vasudevan, V., Vi\'{e}gas, F., Vinyals, O., Warden, P.,
  Wattenberg, M., Wicke, M., Yu, Y., and Zheng, X. (2015).
\newblock {TensorFlow}: Large-scale machine learning on heterogeneous systems.
\newblock Software available from tensorflow.org.

\bibitem[Carpenter et~al., 2017]{stan}
Carpenter, B., Gelman, A., Hoffman, M.~D., Lee, D., Goodrich, B., Betancourt,
  M., Brubaker, M., Guo, J., Li, P., and Riddell, A. (2017).
\newblock Stan: A probabilistic programming language.
\newblock {\em Journal of Statistical Software}, 76(1).

\bibitem[Dickey and Lientz, 1970]{dickey1970}
Dickey, J.~M. and Lientz, B.~P. (1970).
\newblock The weighted likelihood ratio, sharp hypotheses about chances, the
  order of a markov chain.
\newblock {\em The Annals of Mathematical Statistics}, 41(1):214--226.

\bibitem[Faulkenberry, 2018]{faulkenberry2018}
Faulkenberry, T.~J. (2018).
\newblock Computing {Bayes} factors to measure evidence from experiments: An
  extension of the {BIC} approximation.
\newblock {\em Biometrical Letters}, 55(1):31–43.

\bibitem[Gabry and Mahr, 2018]{bayesplot}
Gabry, J. and Mahr, T. (2018).
\newblock {\em bayesplot: Plotting for Bayesian Models}.
\newblock R package version 1.6.0.

\bibitem[Gelfand and Smith, 1990]{gelfand1990}
Gelfand, A.~E. and Smith, A. F.~M. (1990).
\newblock Sampling-based approaches to calculating marginal densities.
\newblock {\em Journal of the American Statistical Association},
  85(410):398--409.

\bibitem[Gigerenzer, 2004]{gigerenzer2004}
Gigerenzer, G. (2004).
\newblock Mindless statistics.
\newblock {\em The {J}ournal of {S}ocio-{E}conomics}, 33(5):587--606.

\bibitem[Gilks et~al., 1994]{gilks1994}
Gilks, W.~R., Thomas, A., and Spiegelhalter, D.~J. (1994).
\newblock A language and program for complex bayesian modelling.
\newblock {\em The Statistician}, 43(1):169--177.

\bibitem[Golding, 2018]{greta}
Golding, N. (2018).
\newblock {\em greta: Simple and Scalable Statistical Modelling in R}.
\newblock R package version 0.3.0.9001.

\bibitem[Hoekstra et~al., 2014]{hoekstra2014}
Hoekstra, R., Morey, R.~D., Rouder, J.~N., and Wagenmakers, E.-J. (2014).
\newblock Robust misinterpretation of confidence intervals.
\newblock {\em Psychonomic {B}ulletin {\&} {R}eview}, 21(5):1157--1164.

\bibitem[Hoel, 1984]{hoel1984}
Hoel, P.~G. (1984).
\newblock {\em Introduction to Mathematical Statistics}.
\newblock John Wiley \& Sons, New York, 5th edition.

\bibitem[{JASP Team}, 2018]{JASP2018}
{JASP Team} (2018).
\newblock {\em {JASP (Version 0.9)[Computer software]}}.

\bibitem[Jeffreys, 1961]{jeffreys1961}
Jeffreys, H. (1961).
\newblock {\em The {T}heory of {P}robability (3rd ed.)}.
\newblock Oxford University Press, Oxford, UK.

\bibitem[Kass and Raftery, 1995]{kass1995}
Kass, R.~E. and Raftery, A.~E. (1995).
\newblock Bayes factors.
\newblock {\em Journal of the American Statistical Association},
  90(430):773--795.

\bibitem[Killeen, 2007]{killeen2007}
Killeen, P.~R. (2007).
\newblock Replication statistics as a replacement for significance testing:
  Best practices in scientific decision-making.
\newblock In Osborne, J.~W., editor, {\em Best practices in quantitative
  methods}, pages 103–--124. SAGE Publications, Inc., Thousand Oaks, CA.

\bibitem[Klugkist et~al., 2005]{klugkist2005}
Klugkist, I., Kato, B., and Hoijtink, H. (2005).
\newblock Bayesian model selection using encompassing priors.
\newblock {\em Statistica Neerlandica}, 59(1):57--69.

\bibitem[Kooperberg, 2018]{polspline}
Kooperberg, C. (2018).
\newblock {\em polspline: Polynomial Spline Routines}.
\newblock R package version 1.1.13.

\bibitem[Kooperberg and Stone, 1992]{kooperberg1992}
Kooperberg, C. and Stone, C.~J. (1992).
\newblock Logspline density estimation for censored data.
\newblock {\em Journal of Computational and Graphical Statistics},
  1(4):301--328.

\bibitem[Lunn et~al., 2000]{bugs}
Lunn, D.~J., Thomas, A., Best, N., and Spiegelhalter, D. (2000).
\newblock {WinBUGS}-a bayesian modelling framework: concepts, structure, and
  extensibility.
\newblock {\em Statistics and computing}, 10(4):325--337.

\bibitem[Masson, 2011]{masson2011}
Masson, M. E.~J. (2011).
\newblock A tutorial on a practical {B}ayesian alternative to null-hypothesis
  significance testing.
\newblock {\em Behavior {R}esearch {M}ethods}, 43(3):679--690.

\bibitem[Morey and Rouder, 2011]{morey2011}
Morey, R.~D. and Rouder, J.~N. (2011).
\newblock Bayes factor approaches for testing interval null hypotheses.
\newblock {\em Psychological {M}ethods}, 16(4):406--419.

\bibitem[Morey and Rouder, 2018]{bayesfactor}
Morey, R.~D. and Rouder, J.~N. (2018).
\newblock {\em BayesFactor: Computation of Bayes Factors for Common Designs}.
\newblock R package version 0.9.12-4.2.

\bibitem[Neal, 2011]{neal2011}
Neal, R. (2011).
\newblock Mcmc using hamiltonian dynamics.
\newblock In Brooks, S., Gelman, A., Jones, G., and Meng, X., editors, {\em
  Handbook of Markov Chain Monte Carlo}, pages 116--162. Chapman and Hall/CRC.

\bibitem[Oakes, 1986]{oakes1986}
Oakes, M. (1986).
\newblock {\em Statistical Inference: {A} commentary for the social and
  behavioural sciences}.
\newblock John Wiley \& Sons, Chicester.

\bibitem[Plummer, 2003]{jags}
Plummer, M. (2003).
\newblock {\em JAGS: A program for analysis of Bayesian graphical models using
  Gibbs sampling}.

\bibitem[{R Core Team}, 2018]{R}
{R Core Team} (2018).
\newblock {\em R: A Language and Environment for Statistical Computing}.
\newblock R Foundation for Statistical Computing, Vienna, Austria.

\bibitem[Raftery, 1995]{raftery1995}
Raftery, A.~E. (1995).
\newblock Bayesian model selection in social research.
\newblock {\em Sociological {M}ethodology}, 25:111--163.

\bibitem[Richard et~al., 2003]{richard2003}
Richard, F.~D., Bond, C.~F., and Stokes-Zoota, J.~J. (2003).
\newblock One hundred years of social psychology quantitatively described.
\newblock {\em Review of General Psychology}, 7(4):331--363.

\bibitem[Rouder et~al., 2009]{rouder2009}
Rouder, J.~N., Speckman, P.~L., Sun, D., Morey, R.~D., and Iverson, G. (2009).
\newblock Bayesian $t$ tests for accepting and rejecting the null hypothesis.
\newblock {\em Psychonomic {B}ulletin {\&} {R}eview}, 16(2):225--237.

\bibitem[Stone et~al., 1997]{stone1997}
Stone, C.~J., Hansen, M.~H., Kooperberg, C., and Truong, Y.~K. (1997).
\newblock Polynomial splines and their tensor products in extended linear
  modeling: 1994 {Wald} memorial lecture.
\newblock {\em The Annals of Statistics}, 25(4):1371--1470.

\bibitem[Wagenmakers et~al., 2011]{wagenmakers2011}
Wagenmakers, E., Wetzels, R., Borsboom, D., and van~der Maas, H. L.~J. (2011).
\newblock Why psychologists must change the way they analyze their data: The
  case of psi: Comment on {B}em (2011).
\newblock {\em Journal of Personality and Social Psychology}, 100(3):426--432.

\bibitem[Wagenmakers et~al., 2010]{wagenmakers2010}
Wagenmakers, E.-J., Lodewyckx, T., Kuriyal, H., and Grasman, R. (2010).
\newblock Bayesian hypothesis testing for psychologists: A tutorial on the
  {S}avage–{D}ickey method.
\newblock {\em Cognitive Psychology}, 60(3):158--189.

\bibitem[Wang, 2017]{wang2017}
Wang, M. (2017).
\newblock Mixtures of $g$ -priors for analysis of variance models with a
  diverging number of parameters.
\newblock {\em Bayesian Analysis}, 12(2):511–532.

\bibitem[Wetzels et~al., 2010]{wetzels2010}
Wetzels, R., Grasman, R.~P., and Wagenmakers, E.-J. (2010).
\newblock An encompassing prior generalization of the savage–dickey density
  ratio.
\newblock {\em Computational Statistics \& Data Analysis}, 54(9):20942102.

\bibitem[Wetzels et~al., 2009]{wetzels2009}
Wetzels, R., Raaijmakers, J. G.~W., Jakab, E., and Wagenmakers, E.-J. (2009).
\newblock How to quantify support for and against the null hypothesis: A
  flexible {WinBUGS} implementation of a default bayesian t test.
\newblock {\em Psychonomic Bulletin \& Review}, 16(4):752--760.

\bibitem[Zellner and Siow, 1980]{zellner1980}
Zellner, A. and Siow, A. (1980).
\newblock Posterior odds ratios for selected regression hypotheses.
\newblock {\em Trabajos de Estadistica Y de Investigacion Operativa},
  31(1):585--603.

\end{thebibliography}

\appendix
\section*{}
In this appendix, I report a simulation study designed to benchmark performance of the posterior sampling generalization of the JZS $t$-test described in this paper against the version of the JZS test originally proposed by \citet{rouder2009}.

For each simulation, I randomly generated 200 single-sample data sets of size $N$ (where $N=20$, 50, or 80) under the model $y_i = \mu + \varepsilon_i$. For each of these data sets, different ``effects'' were represented by varying the parameter $\mu$, which was assumed to be drawn randomly from a normal distribution with mean 0 and variance $g$ \citep[where $g=0$, 0.05, or 0.2; also see][]{wang2017,faulkenberry2018}. Each of the errors $\varepsilon_i$ for a given data set was drawn from a normal distribution with mean 0 and variance 1. The resulting combinations of 3 different sample sizes ($N=20,50,80$) and 3 different effect parameters ($g=0,0.05,0.2$) produced a total of nine simulations.

Once a simulated data set was constructed, I computed the Bayes factor $B_{01}$ for the null hypothesis $\mathcal{H}_0:\mu = 0$ over the alternative hypothesis $\mathcal{H}_1:\mu=1$ using two methods: (1) the JZS Bayes factor of \citet{rouder2009}, and (2) the posterior sampling technique. The JZS Bayes factor was computed using the \verb|ttestBF| function from the BayesFactor package in R, and the posterior sampling Bayes factor was computed using the methods described in this paper in Section 4.1 (i.e., drawing posterior samples using Greta, fitting a logspline estimate of the posterior, and then computing the Savage-Dickey density ratio by comparing the ordinates of $\delta=0$ in the posterior and prior, respectively). Each Bayes factor was computed using a Cauchy prior of scale $r=1$.

In all, I found the two methods to be quite comparable to each other. To see why, let's first inspect the distributions of Bayes factors obtained for each of the nine combinations of $N$ and $g$. Figure \ref{fig:simPlot} shows these via overlaid density plots of $\log(B_{01})$ for each computation method. As one can readily see in Figure \ref{fig:simPlot}, the density plots have considerable overlap, indicating that both methods produced very similar distributions of Bayes factors.

\begin{figure}
  \centering
  \includegraphics[width=0.7\textwidth]{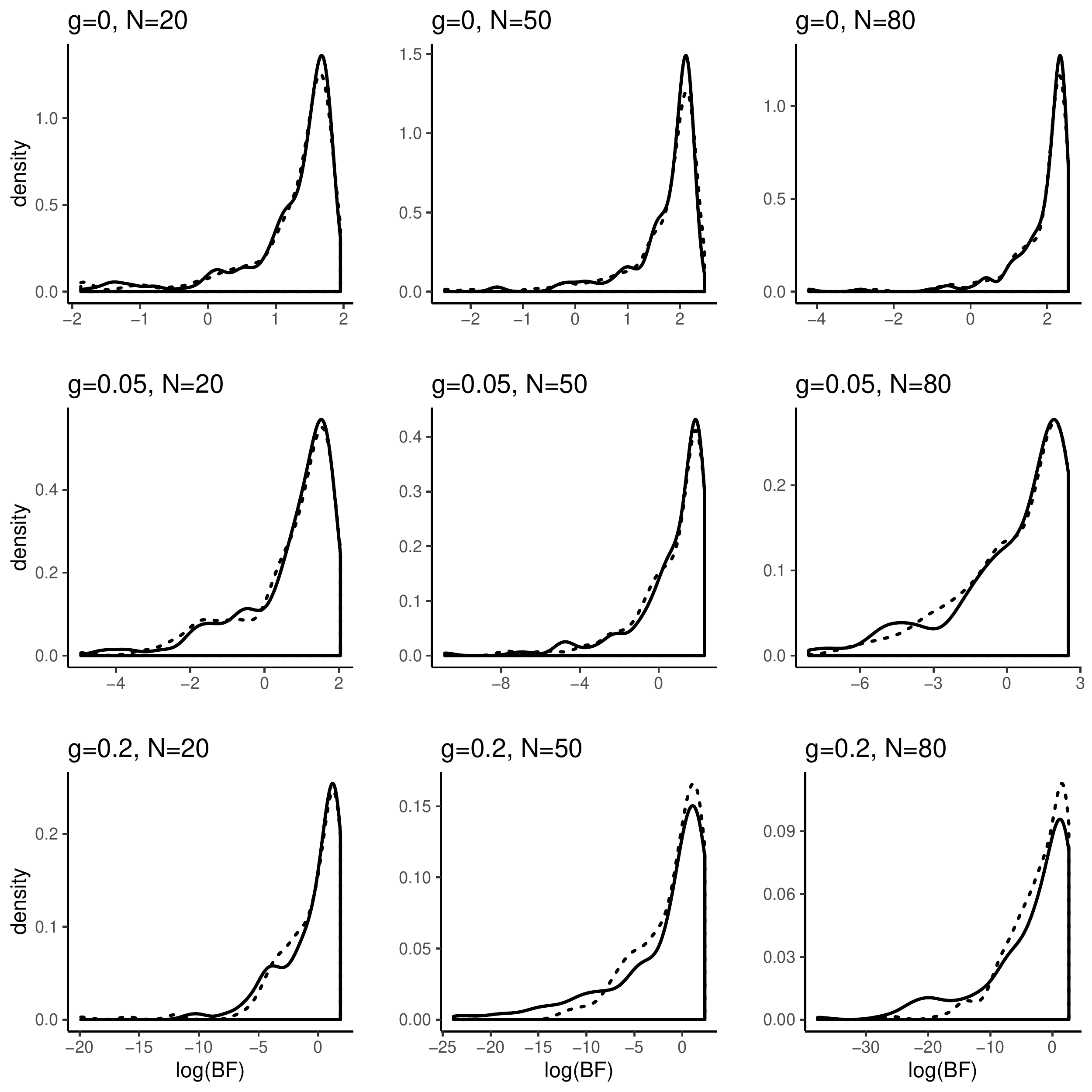}
  \caption{Density plots of $\log(B_{01})$ for each of the nine combinations of "effect" $g$ (0,0.05,0.2) and sample size $N$ (20,50,80). The JZS Bayes factor distribution is displayed as a solid line, whereas the distribution of posterior sampling Bayes factors is displayed as a dashed line.}
  \label{fig:simPlot}
\end{figure}

Further evidence for the compatibility of the two techniques comes from Table \ref{tab:fiveNum}, shows five-number summaries for the values of $\log(B_{01})$ obtained for each condition. Additionally, I computed model selection consistency, defined as the proportion of simulated data sets for which the JZS Bayes factor and the posterior sampling Bayes factor led to the same model choice. For this computation, model selection was determined by computing $\log(B_{01})$. If $\log(B_{01})>0$, then $\mathcal{H}_0$ was selected; otherwise, $\mathcal{H}_1$ was selected \citep[see also][]{faulkenberry2018}. As is shown in Table \ref{tab:fiveNum}, the posterior sampling technique again produced a distribution of Bayes factors that was very similar to those obtained from the JZS Bayes factor, mirroring what is shown in Figure \ref{fig:simPlot}. Critically, both computation methods selected the same model in a large percentage of the simulated data sets, as indicated by the large consistency values in Table \ref{tab:fiveNum}. 

\begin{table}
\footnotesize
\centering
\caption{Five-number summary of $\log(B_{01})$ and model selection consistency.}\label{tab:fiveNum}
\tabcolsep=12pt
\begin{tabular}{ccccccccc}
\hline \hline
  $g$ & $N$ & $BF$ type & Min & $Q_1$ & Median & $Q_3$ & Max & Consistency\\
  0 & 20 & JZS & -1.77 & -1.71 & -1.51 & -1.12 & 1.76 & \\
      &    & sampling & -1.95 & -1.70 & -1.51 &  -1.08 & 1.88 & 0.985\\
      & 50 & JZS & -2.20 & -2.16 & -2.01 & -1.56 & 1.51 &\\
      &    & sampling & -2.47 & -2.16 & -1.99 & -1.53 & 2.49 & 0.985\\
      & 80 & JZS & -2.43 & -2.38 & -2.20 & -1.69 & 4.21 &\\
      &    & sampling & -2.55 & -2.36 & -2.18 & -1.61 & 2.90 & 1.000\\\hline
  0.05 & 20 & JZS & -4.43 & 0.23 & 1.11 & 1.61 & 1.77 & \\
      &     & sampling & -4.95 & 0.27 & 1.12 & 1.62 & 2.04 & 0.995\\
      &  50 & JZS & -10.86 & 0.08 & 1.48 & 1.97 & 2.20 & \\
      &     & sampling & -4.95 & 0.27 & 1.12 & 1.62 & 2.04 & 0.975\\
      &  80 & JZS & -8.09 & -0.79 & 1.14 & 2.06 & 2.43 &\\
      &     & sampling & -6.96 & -0.79 & 1.16 & 2.08 & 2.51 & 0.945\\\hline
  0.2 & 20 & JZS & -10.98 & -1.76 & 0.39 & 1.44 & 1.77 & \\
      &    &  sampling & -19.90 & -1.81 & 0.25 & 1.42 & 1.86 & 0.975\\
      & 50 & JZS & -23.85 & -4.25 & 0.05 & 1.51 & 2.20 & \\
      &    & sampling & -12.39 & -3.19 & 0.01 & 1.49 & 2.36 & 0.965\\
      & 80 & JZS & -37.79 & -6.78 & -1.13 & 1.67 & 2.43 &\\
      &    & sampling & -25.60 & -4.77 & -1.11 & 1.85 & 2.65 & 1.000\\
\hline \hline
\end{tabular}
\end{table}

In all, the proposed sampling method for computing Bayes factors described in this tutorial seems to be quite consistent with the established, albeit less flexible, JZS Bayes factor of \citet{rouder2009}. Thus, the researcher can be confident that the posterior sampling methods described in this paper not only afford a great deal of flexibility, but also benchmark well against other established methods of computation.  

\end{document}